\documentclass[11pt,a4paper]{article}
\usepackage{amsmath,amsfonts,amstext, amsthm, tikz, xspace, cite, todonotes}
\author{Marie-Louise Bruner, Martin Lackner}
\title{A W[1]-Completeness Result for Generalized Permutation Pattern Matching}
\date{\today}

\theoremstyle{plain}
\newtheorem{Thm} {Theorem} [section]

\newtheorem{Def} [Thm]{Definition}

\newtheorem{Cla} {Claim}
\theoremstyle{definition}
\newtheorem{Ex} [Thm]{Example}
\newtheorem*{Ex*}{Example}
\theoremstyle{remark}
\newtheorem{Rem} [Thm]{Remark}

\newenvironment{changemargin}[2]{%
\begin{list}{}{%
\setlength{\topsep}{0pt}%
\setlength{\leftmargin}{#1}%
\setlength{\rightmargin}{#2}%
\setlength{\listparindent}{\parindent}%
\setlength{\itemindent}{\parindent}%
\setlength{\parsep}{\parskip}%
}%
\item[]}{\end{list}}

\newenvironment{Exe*}{\begin{changemargin}{1cm}{0cm}\begin{Ex*}}{\end{Ex*}\end{changemargin}}

\newcommand{\cprob}[3]{
    \begin{center}
      \fbox{
        \parbox{0.9\textwidth}{
          #1\smallskip\\
          \begin{tabular}{rp{0.75\textwidth}}
            \textit{Instance:\ } & #2\\
            \textit{Question:\ } & #3
          \end{tabular}
        }
      }
    \end{center}
}

\newcommand{\pprob}[4]{
    \begin{center}
      \fbox{
        \parbox{0.9\textwidth}{
          #1\smallskip\\
          \begin{tabular}{rp{0.70\textwidth}}
            \textit{Instance:\ } & #2\\
            \textit{Parameter:\ } & #3\\
            \textit{Question:\ } & #4
          \end{tabular}
        }
      }
    \end{center}
}

\newcommand{\pe}{permutation}
\newcommand{\mul}{multiset}
\newcommand{\pea}{permutation }

\newcommand {\cA} {{\mathcal A}}

\newcommand {\bigO} {{\mathcal O}}

\newcommand {\N} {{\mathbb N}}

\newcommand{\ccfont}[1]{\textsf{#1}}
\newcommand{\probfont}[1]{\textsc{#1}}

\newcommand{\ppm}{\probfont{PPM}\xspace}
\newcommand{\gppm}{\probfont{GPPM}\xspace}

\newcommand{\fpt}{\textnormal{\ccfont{FPT}}\xspace}
\newcommand{\xp}{\textnormal{\ccfont{XP}}\xspace}
\newcommand{\w}[1]{\ifmmode{\textnormal{\ccfont{W}\,[#1]}}\else{\textnormal{\ccfont{W}[#1]}}\fi}

\newcommand{\card}[1]{\ensuremath{|#1|}}

\newcommand{\ra}{\rightarrow}

\newcommand{\NP}{\ccfont{NP}\xspace}

\makeatletter
\def\underbracket{%
\@ifnextchar[{\@underbracket}{\@underbracket [\@bracketheight]}%
}
\def\@underbracket[#1]{%
\@ifnextchar[{\@under@bracket[#1]}{\@under@bracket[#1][0.4em]}%
}
\def\@under@bracket[#1][#2]#3{%\message {Underbracket: #1,#2,#3}
\mathop{\vtop{\m@th \ialign {##\crcr $\hfil \displaystyle {#3}\hfil $%
\crcr \noalign {\kern 3\p@ \nointerlineskip }\upbracketfill {#1}{#2}
\crcr \noalign {\kern 3\p@ }}}}\limits}
\def\upbracketfill#1#2{$\m@th \setbox \z@ \hbox {$\braceld$}
\edef\@bracketheight{\the\ht\z@}\bracketend{#1}{#2}
\leaders \vrule \@height #1 \@depth \z@ \hfill
\leaders \vrule \@height #1 \@depth \z@ \hfill \bracketend{#1}{#2}$}
\def\bracketend#1#2{\vrule height #2 width #1\relax}
\makeatother

\begin{document}
\maketitle

\begin{abstract}
The \NP-complete \textsc{Permutation Pattern Matching} problem asks whether a permutation $P$ (the pattern) can be matched into a permutation $T$ (the text).
A matching is an order-preserving embedding of $P$ into $T$.
In the \textsc{Generalized Permutation Pattern Matching} problem one can additionally enforce that certain adjacent elements in the pattern must be mapped to adjacent elements in the text.
This paper studies the parameterized complexity of this more general problem.
We show \w{1}-completeness with respect to the length of the pattern $P$.
Under standard complexity theoretic assumptions this implies that no fixed-parameter tractable algorithm can be found for any parameter depending solely on $P$.
\end{abstract}

\section{Introduction}

\subsection{Permutation patterns}

The concept of pattern avoidance (and, closely related, pattern matching) in \pe s arose in the late 1960ies when Donald Knuth  asked which \pe s could be sorted using one stack in an example of his \textit{Fundamental algorithms} \cite{DBLP:books/aw/Knuth68}. The answer is simple: These are exactly the \pe s avoiding the pattern $231$ and they are counted by the Catalan numbers. By avoiding a certain pattern the following is meant:

\begin{Def}
Let $P$ be a permutation of length $k \leq n$ (the pattern). We say that the $n$-permutation $T$ (the text) \emph{contains $P$ as a pattern} or that \emph{$P$ can be matched into $T$} if we can find a subsequence of $T$ that is order-isomorphic to $P$. 
If there is no such subsequence we say that $T$ \emph{avoids the pattern $P$}.
Matching $P$ into $T$ thus consists in finding a monotone map $\mu: [k] \rightarrow [n]$ so that the sequence $\mu(P)$, being defined as  $\big(\mu(P(i))\big)_{i\in[k]}$, is a subsequence of $T$.
\label{Def.pattern.avoid}
\end{Def}
\begin{Ex}
The text $T=53142$ (being a permutation written in one-line representation) contains the pattern $312$, since the entries $512$ and also $534$ form a $312$-pattern. $T$ however avoids the pattern $123$ since it contains no increasing subsequence of length three. 
\end{Ex}
Since 1985, when Rodica Simion and Frank Schmidt published the first systematic study of \textit{Restricted Permutations} \cite{simion1985restricted}, the area of pattern avoidance in \pe s has become a rapidly growing field of discrete mathematics, more specifically of (enumerative) combinatorics.

There is a far-reaching generalization of pattern matching to so-called \emph{generalized patterns.}
These were introduced in \cite{babson2000generalized} and have since then received a lot of attention (see \cite{steingrımsson376generalized} for a survey).
In contrast to ``classical'' pattern matching, some elements of a generalized pattern may be forced to be mapped to adjacent elements in the text. Let us first give an example:
\begin{Ex}
As in the previous example we consider the text $T=53142$. It contains the generalized pattern $\langle 31\rangle 2$ as shown by the entries $534$.
Observe that the subsequence $\langle 31 \rangle $ of the pattern is mapped to the adjacent elements $53$ in the text.
This would not have been the case for $512$ and hence $512$ does not form a $\langle 31\rangle 2$-pattern.
$T$ avoids the pattern $\langle 312 \rangle$ since it contains no $312$-pattern where all three elements are adjacent.
\end{Ex}
\begin{Def}
A \emph{generalized pattern} $P$ of length $k \leq n$ is a permutation of length $k$ in which parentheses $\langle \rangle$ have been inserted.
These parentheses have to be non-nested, non-overlapping and well-matched, i.e. every parenthesis $\langle$ has to subsequently be followed by a $\rangle$.
Also each pair of parentheses has to contain at least two elements.
A subsequence of $P$ surrounded by a pair of parentheses is called \emph{adjacent block}.

We say that the $n$-permutation $T$ (the text) \emph{contains $P$ as a generalized pattern} or that \emph{$P$ can be matched into $T$} if we can find a subsequence of $T$ that is order-isomorphic to $P$ and in which adjacent blocks are mapped to adjacent elements in the text. 
Matching $P$ into $T$ thus consists in finding a monotone map $\mu: [k] \rightarrow [n]$ so that $\mu(P)$ is a subsequence of $T$ and $\mu(i)$ and $\mu(j)$ are adjacent whenever $i$ and $j$ lie next to each other in an adjacent block.
\label{Def.gen.pattern.avoid}
\end{Def}
We remark that this is not the usual notation which uses dashes instead of parentheses.
For instance $\langle 3 1 \rangle 2$ is usually written as $31-2$, $\langle 3 1 2 \rangle$ as $312$ and $312$ as $3-1-2$.
We prefer to use the parentheses notation since it allows to write classical patterns in the ``classical way''.

\subsection{Computational aspects}

This paper takes the viewpoint of computational complexity.
Computational aspects of pattern avoidance, in particular the analysis of the \textsc{Permutation Pattern Matching} problem, have so far received far less attention than enumerative questions. This problem is defined as follows:

\cprob {\probfont{(Generalized) Permutation Pattern Matching (\gppm/\ppm)}}
{A (generalized) pattern $P$ and a \pea $T$ (the text).}
{Is there a matching of $P$ into $T$?}
\ppm is also known as \probfont{Sub-Permutation} problem in the literature. In \cite{DBLP:conf/wads/BoseBL93} it was shown that the general decision problem is \NP-complete. 
In this paper we focus on the \probfont{Generalized Permutation Pattern 
Matching (\gppm)}.
To the best of our knowledge, so far no work has been done on \gppm.
However it is known to be \NP-complete, since \ppm clearly is a special case.

While knowing about the \NP-completeness of this problem, questions regarding its computational complexity remain unanswered.
A natural question would be ``Is the \probfont{Generalized Permutation Pattern Matching} problem still computationally hard when the pattern is short?''.
%So far, this question has not been answered in the literature.
A framework able to provide an answer is parameterized complexity theory.
In contrast to classical complexity theory, parameterized complexity theory studies the complexity of problems with respect to several parameters and not just the input size.
Therefore an answer to aforementioned question could be found by performing a parameterized complexity analysis using the length $k$ of the pattern as a parameter.

A first step towards an algorithm for \gppm is a trivial $\bigO(n^k)$ brute-force algorithm.
However, even for a moderately long pattern this algorithm is unpractical.
In order to solve \gppm efficiently for short patterns, one could hope for a \emph{fixed-parameter tractable} (fpt) algorithm, i.e. having a runtime of $f(k)\cdot\mathit{poly}(n)$ with $f$ being a computable (ideally single-exponential) function.

In this paper we show that under standard complexity theoretic assumptions,  no such fpt algorithm exists.
This is achieved by showing that \gppm is \w{1}-complete with respect to the length of the pattern. This implies, unless $\fpt=\w{1}$, that there is no fpt algorithm for \gppm for any parameter that depends solely on the pattern.
Furthermore, it follows that an fpt algorithm may only be found if a parameter depending on the text is used.

\section{Related work}
Here we present some relevant work within the complexity analysis of \ppm. 
\begin{itemize}
\item \xp-algorithms, i.e. algorithms with runtime $\bigO(n^{f(k)})$, are given in \cite{DBLP:conf/isaac/AlbertAAH01, DBLP:journals/siamdm/AhalR08}.
\item Special interest has been shown to the case of \textit{separable} patterns. These are permutations that can be represented with the help of a \textit{separating tree} or, equivalently, permutations containing neither $3142$ nor $2413$ as a pattern. Note that every permutation avoiding any one of the patterns $132$, $231$, $213$ and $312$ is separable. Polynomial time algorithms for the decision and counting problems as well as efficient algorithms for the detection of separable permutations can be found in \cite{DBLP:conf/wads/BoseBL93, DBLP:journals/ipl/Ibarra97, DBLP:conf/isaac/AlbertAAH01, DBLP:journals/dam/YugandharS05}.

\item Other cases with restricted patterns have also been studied:
\begin{itemize} 
\item In case $P$ is the identity, \ppm consists of looking for an increasing subsequence of length $k$ in the text which is a special case of the \probfont{Longest Increasing Subsequence} problem. This problem can be solved in $\bigO(n \log n)$-time for sequences in general \cite{schensted1987longest} and in $\mathcal{O}(n \log \log n)$-time for permutations \cite{DBLP:journals/ipl/ChangW92, mäkinen2001longest}.
\item For all patterns of length four \ppm can be solved in $\bigO(n \log n)$-time \cite{DBLP:conf/isaac/AlbertAAH01}.
\item For the case that both the text and the pattern are $321$-avoiding  an $\bigO(k^2n^6)$-time algorithm is presented in \cite{DBLP:conf/isaac/GuillemotV09}. If only the pattern is $321$-avoiding an $\bigO(kn^{4\sqrt{k}+12})$-time algorithm was found.
\end{itemize}
\end{itemize}  
Two other problems related to \ppm are:
\begin{itemize}
\item The \probfont{Longest Common Pattern} problem is to find a longest common pattern between two permutations $T_1$ and $T_2$ i.e. a pattern $P$ of maximal length that can be matched both into $T_1$ and $T_2$. This problem is a generalization of \ppm  since determining whether the longest common pattern between $T_1$ and $T_2$ is $T_1$ is equivalent to \ppm. The \probfont{Longest Common Pattern} problem for the case that one of the two permutations $T_1$ and $T_2$ is separable has been studied e.g. in \cite{bouvel_longest_2006}.

\item For a class of permutations $X$ the  \probfont{Longest $X$-Subsequence} (LXS) problem is to identify in a given permutation $T$ its longest subsequence that is isomorphic to a permutation of $X$. 
%Note that LXS reduces to LIS if $X=I$ is the set of identity permutations of all lengths. 
In \cite{albert2003longest} polynomial time algorithms for many classes $X$ are described, in general however LXS is \NP-hard.
\end{itemize}
As mentioned before, we do not know of any work that has been done on \gppm.

\section{Preliminaries}

We give the relevant definitions of parameterized complexity theory. Details can be found in \cite{DowneyF99Book, FlumG2006parameterized}.
In contrast to classical complexity theory, a parameterized complexity analysis studies the runtime of an algorithm with respect to several parameters and not just the input size~$\card{I}$. Therefore every parameterized problem is a subset of $\Sigma^*\times \N$, where $\Sigma$ is the input alphabet. An instance of a parameterized problem consequently consists of an input string together with a positive integer $k$, the parameter. 
\begin{Def}
A parameterized problem $P$ is \emph{fixed-parameter tractable} (or in \fpt) if there is a computable function $f$ and a polynomial $p$ such that 
there is an algorithm solving $P$ in time $f(k)\cdot p(\card{I})$.
Such an algorithm is called fixed-parameter tractable as well.
\end{Def}
We continue by defining parameterized reductions.
\begin{Def}
Let $L_1,L_2\subseteq \Sigma^*\times \N$ be two parameterized problems.
An \emph{fpt-reduction} from $L_1$ to $L_2$ is a mapping $R : \Sigma^*\times \N \ra \Sigma^*\times \N$ such that
\begin{enumerate}
\item $(I, k) \in L_1$ iff $R(I, k) \in L_2$.
\item $R$ is fixed-parameter tractable.
\item There is a computable function g such that for $R(I, k) = (I' , k')$, $k' \leq g(k)$ holds.
\end{enumerate}
\end{Def}

\noindent \fpt is the parameterized equivalent of \ccfont{PTIME}. 
Other important complexity classes in the framework of parameterized complexity are those in the \textnormal{\ccfont{W}}-hierarchy, $\w{1}\subseteq\w{2}\subseteq\ldots$. 
For our purpose only the class \w{1} is relevant.
It is conjectured (and widely believed) that $\w{1}\neq \fpt$ (see e.g. \cite{DantchevMS11}). 
Therefore showing \w{1}-hardness can be considered as evidence that the problem is not fixed-parameter tractable.

\begin{Def}
The class $\w{1}$ is defined as the class of all problems that are fpt-reducible to the following problem.
\end{Def}
\pprob
{\probfont{Turing Machine Acceptance}}
{A nondeterministic Turing machine with its transition table,
an input word $x$ and a positive integer $k$.}
{$k$}
{Does the Turing machine accept the input $x$ in at most $k$ steps?}

\begin{Def}
A parameterized problem $P$ is in \xp if there is a computable function $f$ such that 
there is an algorithm solving $P$ in time $\bigO(\card{I}^{f(k)})$.
\end{Def}

\noindent All the aforementioned classes are closed under fpt-reductions. 
The following relations between these complexity classes are known:
\[
\fpt\subseteq\w{1}\subseteq\w{2}\subseteq\ldots\subseteq\xp
\]

\noindent Finally, $[k]$ denotes the set $\{1,\dots,k\}$.

\section{The \w{1}-completeness result}

\begin{Thm}
\gppm is \w{1}-complete with respect to length of the pattern.
\label{thm:w1_completeness}
\end{Thm}

\begin{proof} 

We show W[1]-hardness by giving an fpt-reduction from the following problem to \gppm:

\pprob{\textsc{Independent Set}} {A graph $G=(V, E)$ and a positive integer $k$.}{$k$}{Is there a subset $S \subseteq V$ of size $k$, so that the induced subgraph $G[S]$ contains no edges?}

Let $(G,k)$ be an \textsc{Independent Set} instance, where $V=\left\lbrace v_1, v_2, \ldots v_l \right\rbrace $ is the set of vertices and $E=\left\lbrace e_1, e_2, \ldots, e_m \right\rbrace $ the set of edges. We are going to construct a \gppm instance $(P, T)$. We shall first construct a pair $(P', T')$. $P'$ is a generalized multiset pattern, i.e.\ a generalized pattern in which elements may occur more than once. $T'$ is a \pea on a \mul. Applying Definition~\ref{Def.gen.pattern.avoid} to \pe s on \mul s means that in a matching repeated elements in the pattern have to be mapped to repeated elements in the text. Afterwards we deal with the repeated elements in order to create a generalized pattern and a \pea on ordinary sets and hereby obtain $(P, T)$. 

Both the pattern and the text consist of a substring coding vertices ($\dot P$ resp.\ $\dot T$) and a substring coding edges ($\bar P$ resp.\ $\bar T$). In between these two substrings we place a \emph{separator block} of constant length $c$ to ensure that $\dot P$ is matched into  $\dot T$ and $\bar P$ into $\bar T$. For the moment, we will simply write $\|$ for this separator block, indicating that the separator block in the pattern has to be mapped to the one in the text. Its construction shall be described later on.

We define the pattern to be 
\begin{align*}
P'  &=  \dot P \| \bar{P} =  1 2 3 \ldots k \| 	\langle 121\rangle \langle 131 \rangle \ldots \langle 1k1 \rangle \langle 232 \rangle \ldots \langle2k2 \rangle \ldots \langle(k-1)k(k-1)\rangle
\end{align*}
$\dot P$ corresponds to a list of (indices of) $k$ vertices. $\bar P$ represents all possible edges between the $k$ vertices (in lexicographic order). 
An edge from $i$ to $j$ is encoded by $\langle iji \rangle$. 
Since the complete graph on $k$ vertices has ${k(k-1)}/{2}$ edges the total length of $P$ is equal to $k+c+3k(k-1)/2= c+(3k^2-k)/2$. 
\\
For the text $T'  = \dot T \| \bar T$ we proceed similarly. $\dot T$ is a list of the (indices of the) $l$ vertices of $G$. $\bar T$ represents all edges \textit{not occurring} in $G$ listed in lexicographic order. An edge $\{i,j\}$ is again encoded by $iji$. Let us give an example:

\begin{Exe*}%[IS to \ppm-reduction]
Let $l=6$ and $k=3$. Then the pattern is given by \[P'=123 \| \langle121\rangle \langle131\rangle \langle232\rangle\]
Consider for instance the graph $G$ with six vertices $v_1, \ldots , v_6$ and edge-set 
\[\left\lbrace 
 \left\lbrace 1, 3\right\rbrace ,   \left\lbrace 1, 4\right\rbrace ,  \left\lbrace 1, 5\right\rbrace ,  \left\lbrace 2, 6\right\rbrace , \left\lbrace 3, 4\right\rbrace , \left\lbrace 3, 6\right\rbrace , \left\lbrace 5, 6\right\rbrace \right\rbrace \]
represented in Figure \ref{example_w1_proof} (we write $\left\lbrace i, j\right\rbrace $ instead of $\left\lbrace v_i, v_j\right\rbrace $). Then the eight edges not appearing in $G$ are 
\[\left\lbrace 
 \left\lbrace 1, 2\right\rbrace ,   \left\lbrace 1, 6\right\rbrace , \left\lbrace 2, 3\right\rbrace , \left\lbrace 2, 4\right\rbrace ,  \left\lbrace 2, 5\right\rbrace , \left\lbrace 3, 5\right\rbrace ,  \left\lbrace 4, 5\right\rbrace , \left\lbrace 4, 6\right\rbrace \right\rbrace. \]
The text is thus given by: \[T'= 123456 \| 121 \; 161 \; 232  \; 242 \;  252  \; 353  \; 454 \;  464. \]

\begin{figure}
\begin{center}
\begin{tikzpicture}%[scale=0.8, thick]
  \tikzstyle{gray_vertex}=[circle,draw=black,fill=black!25,minimum size=17pt,inner sep=0pt]
  \tikzstyle{red_vertex}=[circle,draw=black,fill=red,minimum size=17pt,inner sep=0pt]
  \tikzstyle{red_box}=[rectangle,%draw=black,
  fill=red,minimum size=12pt,inner sep=0pt]

% Graph mit markiertem independent set mit drei knoten
\foreach \name/\angle/\text in { P-1/60/1, P-6/120/6,  P-4/240/4}
    \node[gray_vertex,xshift=6cm,yshift=.5cm] (\name) at (\angle:2cm) {$v_{\text}$};
\foreach \name/\angle/\text in {P-2/0/2, P-5/180/5,  P-3/300/3}
    \node[red_vertex,xshift=6cm,yshift=.5cm] (\name) at (\angle:2cm) {$v_{\text}$};

\foreach \from/\to in {1/3, 1/4, 1/5, 2/6, 3/4, 3/6, 5/6}
    { \draw[thick] (P-\from) -- (P-\to);  }

%pattern
\node at (9.0, 1.25) {$P$:};
\foreach \x/\y/\name/\text in {10.75/1/1/1, 11.25/1/2/2, 11.75/1/3/3, 14.3/0.95/12/$\langle 121 \rangle$, 15.7/0.95/13/$\langle 131 \rangle$, 17.1/0.95/23/$\langle 	232 \rangle$}
    \node (A-\name) at (\x-0.5, \y+0.25) {\text};

%Text mit markiertem pattern
\node at (9.0, -0.25) {$T$:};
\foreach \x/\y/\name in {10/0/1, 11.5/0/4,  12.5/0/6,  13.5/0/121, 14.2/0/161, 15.6/0/242, 17.7/0/454, 18.4/0/464}
    \node (B-\name) at (\x-0.5, \y-0.25) {\name};
\foreach \x/\y/\name in { 10.5/0/2, 11/0/3, 12/0/5, 14.9/0/232, 16.3/0/252, 17/0/353}
    \node[red_box] (B-\name) at (\x-0.5, \y-0.25) {\name};
    
%separators
\node at (12.5, 1.25) {$\|$};
\node at (12.5, -0.25) {$\|$};
    
%Matching pattern in Text
 \foreach \from/\to in {1/2, 2/3, 3/5, 12/232, 13/252, 23/353}
    { \draw[shorten >=1pt, ->] (A-\from) -- (B-\to);  }
    
\end{tikzpicture}
\end{center}
\caption{An example for the reduction of an \probfont{Independent Set} instance to a \ppm instance.}
\label{example_w1_proof}
\end{figure}
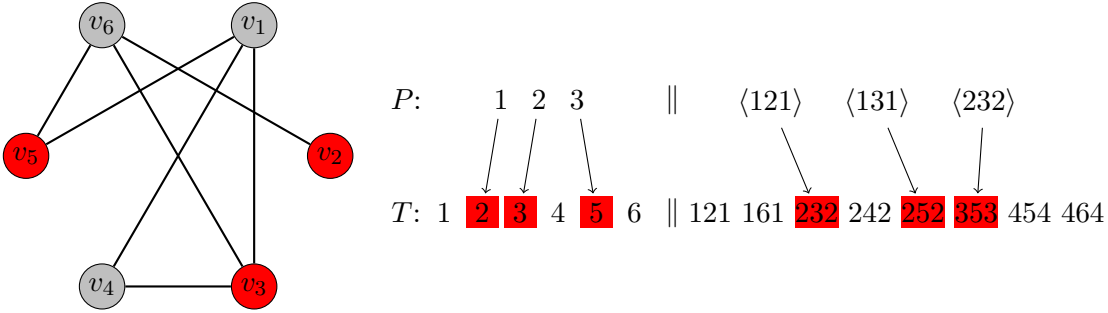
\end{Exe*}

\noindent Since there are at most ${l(l-1)}/{2}$ edges not appearing in $G$, the length of $T'$ is at the most $c+(3l^2-l)/2$. 

\begin{Cla}
An independent set of size $k$ can be found in $G$ if and only if there is a simultaneous matching of $\dot P$ into $\dot T$ and of $\bar P$ into $\bar T$.
\label{claim.simult.matching}
\end{Cla}

\begin{Exe*}[continuation]
In our example $\left\lbrace v_2, v_3, v_5 \right\rbrace $ is an independent set of size three. Indeed, the pattern $P'$ can be matched into $T'$ as can be seen by matching the elements $1, 2 \text{ and } 3$ onto $2, 3 \text{ and } 5$ respectively. See again Figure \ref{example_w1_proof} where the involved vertices respectively elements of the text have been marked in red. 
\end{Exe*}

\begin{proof}[Proof of Claim \ref{claim.simult.matching}]
A matching of $\dot P$ into $\dot T$ corresponds to a selection of $k$ vertices amongst the $l$ vertices of $G$. If it is possible to additionally match $\bar P$ into $\bar T$ this means that there are no edges between the selected vertices. 
Observe that a matching of the generalized pattern $\bar P$ into $\bar T$ consists of mapping each edge between selected vertices to an edge not appearing in $G$. This is because every adjacent triple of the form $iji$ (with $i<j$) in $\bar T$ corresponds to one of the non-edges of the graph.
The selected $k$ vertices thus form an independent set in $G$. 
Conversely, if for every possible matching of $\dot P$ into $\dot T$ defined by a monotone map $\mu: [k] \rightarrow [l] $ some $\langle xyx \rangle$ in $\bar P$ cannot be matched  into $\bar T$, this means that $\left\lbrace \mu(x), \mu(y)\right\rbrace $ is one of the edges of $G$ since $\mu(x) \mu(y) \mu(x)$ does not appear as adjacent elements in $\bar T$. Thus, for every selection of $k$ vertices there will always be at least one edge connecting two of them and therefore there is no independent set of size $k$ in $G$. 
\end{proof}

In order to get rid of repeated elements, we identify every variable with a real interval: $1$ corresponds to the interval $[1,1.9]$, $2$ to $[2,2.9]$ and so on until finally $k$ corresponds to $[k,k+0.9]$ (resp.\ $l$ to $[l,l+0.9]$). In $\dot P$ and $\dot T$ we shall therefore replace every element $j$ by the pair of elements $(j+0.9,j)$ (in this order). The occurrences of $j$ in $\bar P$ (resp.\ $\bar T$) shall then successively be replaced by real numbers in the interval $[j, j+0.9]$. For every $j$, these values are chosen one after the other (from left to right), always picking a real number that is larger than all the previously chosen ones in the interval $[j, j+0.9]$. Note that this construction changes the length of the pattern (resp.\ the text) by $k$ (resp.\ $l$) since the length of $\dot P$ (resp.\ $\dot T$) is doubled. We now have $|P|=c+(3k^2+k)/2$ and $|T| \leq c + (3l^2+l)/2$.

Observe the following: The obtained sequence is not a \pea in the classical sense since it consists of real numbers. However, by replacing the smallest number by 1, the second smallest by 2 and so on, we do obtain an ordinary \pe . This defines $P$ and $T$ (except for the seperator block). 
%In total, this needs $(3k^2+k)/2$ distinct integers (the elements of the seperator block not yet considered) for $P$ and at the most $(3l^2+l)/2$ for $T$. Also 

\begin{Exe*}[continuation]
Getting rid of repetitions in the pattern of the above example could for instance be done in the following way:
\[P= 1.9 \; 1 \; 2.9 \;  2 \; 3.9 \;  3 \| \langle 1.1 \;  2.1 \; 1.2 \rangle \langle 1.3\;  3.1 \; 1.4 \rangle \langle 2.2 \; 3.2 \; 2.3\rangle \]
This \pea of real numbers is order-isomorphic to the following ordinary \pe :
\[P=61\text{B}7\text{F}\text{C} \| \langle283\rangle  \langle4\text{D}5\rangle \langle9\text{E}\text{A}\rangle\]
For increased legibility, we use letters A, B, C,... with lexicographic order for numbers larger than $9$, i.e. $\text{A}=10, \text{B}=11,$ and so on. 
\end{Exe*}

\begin{Cla}
$P$ can be matched into $T$ iff $P'$ can be matched into $T'$.
\label{claim.repetitions}
\end{Cla}

\begin{proof}[Proof of Claim \ref{claim.repetitions}]
Suppose that $P'$ can be matched into $T'$. When matching $P$ into $T$, we have to make sure that elements in $P$ that were copies of some repeated element in $P'$ may still be mapped to elements in $T$ that were copies themselves in $T'$. Indeed this is possible since we have chosen the real numbers replacing repeated elements in increasing order. If $i$ in $P'$ was matched to $j$ in $T'$, then the pair $(i+0.9,i)$ in $P$ may be matched to the pair $(j+0.9,j)$ in $T$ and the increasing sequence of elements in the interval $[i,i+0.9]$ may be matched into the increasing sequence of elements in the interval $[j,j+0.9]$.

Now suppose that $P$ can be matched into $T$. In order to prove that this implies that $P'$ can be matched into $T'$, we merely need to show that elements in $P$ that were copies of some repeated element in $P'$ have to be mapped to elements in $T$ that were copies themselves in $T'$. Then returning to repeated elements clearly preserves the matching.
Firstly, it is clear that a pair of consecutive elements $i+0.9$ and $i$ in $\dot P$ has to be matched to some pair of consecutive elements $j+0.9$ and $j$ in $\dot T$, since $j$ is the only element smaller than $j+0.9$ and appearing to its right. Thus intervals are matched to intervals. Secondly, an element $x$ in $P$ for which it holds that $i < x < i +0.9$ must be matched to an element $y$ in $T$ for which it holds that $j < y < j +0.9$. Thus copies of an element are still matched to copies of some other element. 

Finally, replacing real numbers by integers clearly does not change the permutations in any relevant way.
\end{proof}

We now have to implement the separator block in order to ensure that the two substrings $\dot P$ and $\bar P$ of the pattern are matched into the corresponding substrings $\dot T$ resp. $\bar T$ in $T$. 
In $P$ we proceed in the following way. Let us denote its largest element by $P_{\max}$ (this is the next-to-last element of $\dot P$ and was denoted by $(k+0.9)$ in $P'$). Then we define the separator block to be the generalized pattern $\langle P_{\max}+3, P_{\max}+2, P_{\max} +1, P_{\max}+4 \rangle$. This leads to $P=\dot P  \langle P_{\max}+3, P_{\max}+2, P_{\max} +1, P_{\max}+4 \rangle \bar{P} $. For $T$ we proceed similarly and obtain $T=\dot T  (T_{\max}+3) (T_{\max}+2)(T_{\max} +1) (T_{\max}+4) \bar{T}$. We thus have $c=4$, i.e.\ the separator block has length $4$.

\begin{Exe*}[continuation]
The largest element of $P$ is $\text{F}=15$, thus we will use the following three elements as separator block: G, H, I and J. Inserting these elements as described above leads to:
\[P=61\text{B}7\text{F}\text{C} \langle \text{IHGJ} \rangle \langle283\rangle  \langle4\text{D}5\rangle \langle9\text{E}\text{A}\rangle\]
\end{Exe*}

\begin{Cla}
Constructing the separator block in the described way guarantees that in a matching of $P$ into $T$, $\dot P$ is matched into $\dot T$ and $\bar P$ is matched into $\bar T$.
\label{claim.guards}
\end{Cla}

\begin{proof}[Proof of Claim \ref{claim.guards}]
The only $\langle 3214 \rangle$-pattern that can be found in $T$, i.e. the only decreasing subsequence consisting of three adjacent elements followed by a larger element, is formed by the elements $(T_{\max}+3) (T_{\max}+2)(T_{\max} +1) (T_{\max}+4)$. Thus the separator block in $P$ must be mapped to the one in $T$. Consequently the elements lying to its left (these are exactly the elements of $\dot P$) must be matched into the elements lying to the left of the separator block in $T$, i.e. into $\dot T$. For the same reason $\bar P$ must be matched into $\bar T$.
 \end{proof}

\noindent This finally yields that $(G,k)$ is a YES-instance of \probfont{Independet Set} if and only if $(P,T)$ is a YES-instance of \gppm. It remains to show that this reduction can be done in fpt-time. We have already remarked that the length of the pattern is equal to $4+(3k^2+k)/2$ and the length of the text is at the most $4+(3l^2+l)/2$. Thus $|P|=\mathcal{O}(k^2)$ and $|T|=\mathcal{O}(l^2)$.

\medskip
For showing membership we encode \gppm as a model checking problem of an existential first order formula.
\w{1}-membership is then a consequence of the fact that the following problem is \w{1}-complete~\cite{FlumG05}.
\pprob
{\probfont{Existential first-order model checking}}
{A structure $\cA$ and an existential first-order formula $\varphi$}
{$\card{\varphi}$}
{Is $\cA$ a model for $\varphi$?}
\noindent Let $k=\card{P}$.
We compute a structure $\cA=(A,<,T_<,S)$, where the domain set $A=\{1,\ldots,n\}$ represents indices in the text. 
$T_<$ is a binary relation where $T_<(x,y)$ is true iff the element on the $x$-th position in $T$ is less than the element on position $y$.
$S$ is a binary relation where $S(x,y)$ is true iff $y$ is the successor of $x$, i.e. $x+1=y$. 
$T_<$, $S$ and $<$ can be computed in polynomial time. The formula $\varphi$ we want to check is
\begin{align*}
\varphi  = & \exists x_1 \ldots \exists x_k \quad {x_1<x_2} \ \wedge \   {x_2<x_3} \  \wedge \  \ldots  \  \wedge \   x_{k-1}<x_k \ \wedge\\  & \bigwedge_{\substack{P(i)<P(j) \\ \text{for }i,j\in [k]}} T_<(x_i,x_j) \wedge \bigwedge_{\substack{P(i)>P(j) \\ \text{for }i,j\in [k]}} \neg T_<(x_i,x_j) \wedge \bigwedge_{\substack{AdjToRight(i)\\\text{for }i\in[k-1]}} S(x_i,x_{i+1}),
\end{align*}
where $AdjToRight(i)$ is true iff $P(i)P(i+1)$ is contained in an adjacent block. 
Observe that the length of $\varphi$ is in $\bigO(k^2)$. The correctness of the reduction follows directly from Definition~\ref{Def.gen.pattern.avoid}. Indeed, the fact that $P$ can be matched into $T$ means that indices $x_1, \ldots, x_k $ where $x_1< \ldots < x_k$ can be found so that $P(i)<P(j)$ iff $T(x_i)<T(x_j)$ and so that $x_i+1=x_{i+1}$ holds whenever $AdjToRight(i)$ is true.
\end{proof}

\begin{Rem}
Since in the fpt-reduction the length of the pattern can be bounded by a polynomial in the size of $G$, this is also a polynomial time reduction. 
Therefore the proof of Theorem~\ref{thm:w1_completeness} can also be seen as an alternative way of showing \NP-hardness for \gppm. This also follows from the \NP-hardness of the less general \ppm.
\end{Rem}

\section{Conclusion}

We have shown that the \probfont{Generalized Permutation Pattern Matching} problem is \w{1}-complete with respect to the length of the pattern.
This implies that under standard complexity theoretic assumptions there is no fpt-al\-go\-ri\-thm for this problem with respect to the length of the pattern.
Furthermore, this also yields that \emph{no} parameter that is a function of the pattern can yield fixed-parameter tractability.
In order to obtain fpt results, we plan to study parameters concerning the text.
The question whether \ppm is also \w{1}-hard remains open.

\section{Acknowledgements}
We would like to thank Marek Cygan, Marcin Pilipczuk and Ond\v{r}ej Such\'{y} for pointing out a flaw in the proof of a previous version of this paper.

\bibliographystyle{plain}
\bibliography{lit}

\end{document}